\def\finalversion{1}

\documentclass[11pt]{article}
\usepackage[utf8]{inputenc}
\usepackage[T1]{fontenc}
\usepackage[sc]{mathpazo}
\usepackage{microtype}
\usepackage{amsmath,xcolor}
\usepackage{amssymb}
\usepackage{amsthm}
\usepackage{bm,float}
\usepackage{dsfont}
\usepackage{authblk}
\usepackage{fullpage,caption,wrapfig}
\usepackage{comment}
\usepackage{mathtools}
\usepackage{nicefrac}
\usepackage[shortlabels]{enumitem}
\usepackage[backend=bibtex, style=alphabetic, backref=true,maxbibnames=99, url=false]{biblatex} 
\usepackage{bbm}

\usepackage{nag,tikz}
\usetikzlibrary{calc}
\usetikzlibrary{decorations.pathreplacing}
\usetikzlibrary{shapes, patterns, decorations, fit, intersections, arrows, automata, positioning}

\usepackage{algorithm,algpseudocode}

\usepackage{multicol}
\usepackage[scaled]{helvet} 
\usepackage{thmtools} 
\usepackage{hyperref}
\hypersetup{
    colorlinks=true,
    linkcolor=violet,
    filecolor=magenta,      
    urlcolor=cyan,
    citecolor=blue,
    pdffitwindow=true,
}
\usepackage[capitalize, nameinlink]{cleveref}
\makeatletter
\AddToHook{cmd/appendix/before}{\def\cref@section@alias{appendix}}
\makeatother

\theoremstyle{plain}
\newtheorem{theorem}{Theorem}[section]
\newtheorem{lemma}[theorem]{Lemma}
\newtheorem{corollary}[theorem]{Corollary}

\newtheorem{fact}[theorem]{Fact}

\newtheorem{claim}[theorem]{Claim}

\theoremstyle{definition}
\newtheorem{definition}[theorem]{Definition}

\theoremstyle{remark}

\newenvironment{nestedproof}{\begin{proof}}{\end{proof}}

\newcommand{\old}[1]{}
\def\cR{{\cal R}}

\newcommand{\poly}{\text{poly}}

\newcommand{\Cons}{\textsc{MakeLam}}
\newcommand{\ltarget}{l_S}
\newcommand{\rtarget}{r_{S}}
\newcommand{\lone}{i}
\newcommand{\rone}{j}
\newcommand{\lactive}[1]{l_{#1}}
\newcommand{\ractive}[1]{r_{#1}}
\newcommand{\rlast}{z}

\newcommand{\onefifth}{\nicefrac{1}{5}}
\newcommand{\onefortieth}{\nicefrac{1}{40}}
\usepackage{stmaryrd}
\newcommand{\out}[2]{\langle #1, #2\rangle}
\newcommand{\cut}[2]{\ensuremath{\mathrm{cut}({#1}, {#2})}}

\newcommand{\Cover}{\cB}

\newcommand{\depth}{\mathit{depth}}
\newcommand{\untouched}{\mathit{untouched}}

\newcommand{\depthact}{\depth}
\newcommand{\untouchedact}{\untouched}

\newcommand{\R}{\ensuremath{\mathbb R}}

\renewcommand{\path}[2]{{ S_{#1}, \ldots, S_{#2} }}

\def\b1{{\bf 1}}
\def\1{{\bf 1}}

\def\cB{{\cal B}}

\def\cD{{\cal D}}

\def\cA{{\cal A}}

\def\cL{{\cal L}}

\def\setminus{\smallsetminus}
\def\R{\mathbb{R}}
\def\cR{{\mathcal{R}}}

\def\cC{{\cal C}}

\newcommand{\declareperson}[1]{\expandafter\newcommand\csname#1\endcsname[1]{\textcolor{orange}{#1: ##1}}}

\declareperson{Nathan}
\declareperson{zs}
\newcommand{\nnote}[1]{\textcolor{purple!70!black}{\em Neil: #1}}

\ifdefined\finalversion
    \renewcommand{\nnote}[1]{}
    \renewcommand{\Nathan}[1]{}
    \renewcommand{\zs}[1]{}
    
\fi

\hypersetup{ colorlinks=true, linkcolor=blue, filecolor=magenta, urlcolor=blue, }

\definecolor{arylideyellow}{rgb}{0.91, 0.84, 0.42}

\begin{document}

\title{Thin Trees for Near Minimum Cuts\thanks{The first and third authors were supported by the NSF CAREER grant CCF-2442250. The second author was supported by the NWO Vidi grant 016.Vidi.189.087.}}
\author[1]{Nathan Klein}
\affil[1]{Boston University}
\author[2]{Neil Olver}
\affil[2]{London School of Economics
 and Political Science}
\author[1]{Zi Song Yeoh}
\date{} 
\maketitle

\begin{abstract}
    The strong thin tree conjecture states that every $k$-edge-connected graph $G$ contains an $O(1/k)$-thin spanning tree, meaning a spanning tree which contains at most an $O(1/k)$ fraction of the edges across each cut in $G$. This conjecture is still open despite significant effort; the best current result by Anari and Oveis Gharan shows the existence of an $O(\poly\log\log n/k)$-thin tree. 

    In this work, we demonstrate that the conjecture is true if one only requires thinness for the set of $\eta$-near minimum cuts of the graph for $\eta = \onefortieth$, in other words, for the set of cuts with fewer than $(1+\onefortieth)k$ edges. Our approach constructs such a tree in polynomial time. To show this, we utilize the structure of near minimum cuts, and in particular the polygon representation of Bencz{\'u}r and Goemans, to reduce to the previously solved problem of finding a spanning tree that is $O(1/k)$-thin for all sets in a laminar family.
\end{abstract}

\section{Introduction}

Given a graph $G=(V,E)$, a spanning tree $T$ of $G$ is called \emph{$\alpha$-thin} if for every cut $S \subsetneq V$ we have the inequality $|T \cap \delta(S)| \le \alpha |\delta(S)|$. In other words, for each cut of $G$, $T$ must have at most an $\alpha$-fraction of the edges.

Goddyn~\cite{God04} conjectured that there exists a function $f$ with $\lim_{k \to \infty} f(k) = 0$ so that every $k$-edge-connected graph $G$ has an $f(k)$-thin spanning tree. Put another way, the conjecture is that for every $\alpha \in (0,1]$, there exists some $k$ so that every $k$-edge-connected graph has an $\alpha$-thin tree. This is known as the thin tree conjecture. \cite{AGMOS17} conjectured that one can pick $f(k) = C/k$ for some constant $k$, which is the best one could hope for up to constant factors, as the minimum cut has $k$ edges and any spanning tree must have at least one edge across the cut. This stronger version is known as the \emph{strong thin tree conjecture}, and was shown by \cite{AGMOS17} to imply an $O(1)$ integrality gap for the standard LP relaxation for the asymmetric traveling salesperson problem (ATSP). An $O(1)$ integrality gap and approximation factor for ATSP was obtained by Svensson, Tarnawski, and V{\'e}gh in a breakthrough result~\cite{STV18} (and further improved by \cite{TV20ATSP}), but this would give an alternative approach to the problem. 
The thin tree conjecture would also imply an alternative proof of the weak 3-flow conjecture \cite{Jae84}, which was resolved by Thomassen~\cite{Thom12}.

The thin tree conjecture is still open; the best known result is $O(\poly\log\log n /k)$ thinness~\cite{AO15}. 
One challenge is the exponential number of cuts involved.
So a natural direction to explore is to ask for thinness only for some cuts: given a collection of cuts $\mathcal{C}$, say that a spanning tree $T$ is $\alpha$-thin with respect to $\mathcal{C}$ if $|T \cap \delta(S)| \leq \alpha|\delta(S)|$ for all $S \in \mathcal{C}$.
Even here, results are limited. 
If $\mathcal{C}$ is a laminar family\footnote{Recall a family $\cL$ is laminar if for all $S,T \in \cL$ we have $S \cap T \in \{\emptyset, S ,T\}$.}, 
Klein and Olver~\cite{KO23} showed that an $O(1)$-thin spanning tree with respect to $\cC$ exists. 

We give a positive result for another natural structured family.
Given a $k$-edge-connected graph, a cut $S \subseteq V$ is an \emph{$\eta$-near minimum cut}\footnote{Our use of a strict inequality here follows the convention of Bencz{\'ur} \cite{Ben95}.} if $|\delta(S)| < (1+\eta)k$.
We show that the strong thin tree conjecture is true for the set of near minimum cuts. In particular, we prove the following theorem:
\begin{restatable}{theorem}{mainthm}\label{thm:main-result}
Let $G=(V,E)$ be a $k$-edge-connected graph. Then for $\eta = \onefortieth$, there is a spanning tree $T$ with $|T \cap \delta(S)| \le 88$ for all cuts $S$ with $|\delta(S)| < (1+\eta)k$ (and therefore, $|T \cap \delta(S)| \le \frac{88}{k} \cdot |\delta(S)|$). Furthermore, we can find such a tree in polynomial time. 
\end{restatable}
We remark that the collection of near minimum cuts can be as large as $\Omega(|V|^2)$ (on the cycle graph); by contrast, a laminar family has size at most $2|V|-1$.

Another motivation for considering the family of near minimum cuts comes from a conjecture due to Pritchard~\cite{Pri11} concerning removable spanning trees. 
Pritchard's conjecture states that there exists a function $g(k)$ with $\lim_{k \to \infty} g(k)/k \to 0$, such that
every $k$-edge-connected graph $G=(V,E)$ contains a spanning tree $T$ so that the graph obtained after deleting $T$ is $k-g(k)$ edge connected.
The strong version of Pritchard's conjecture states that one can choose $g(k)$ to be an absolute constant.
The best current bound for this conjecture is that there is a tree we can delete so that the resulting graph is $\lfloor \frac{k}{2} \rfloor-1$ edge connected, and is an immediate consequence of the Nash-Williams theorem \cite{NW61}.

Pritchard's conjecture is implied by the thin tree conjecture (and its strong form is implied by the strong thin tree conjecture). Indeed, if $T$ is $f(k)$-thin, then the connectivity of the graph obtained by removing $T$ is at least
\[ \min_{\emptyset \subsetneq S \subsetneq V} |\delta(S)| - |T \cap \delta(S)| \geq 
\min_{\emptyset \subsetneq S \subsetneq V} |\delta(S)|(1 - f(k)) \geq k(1-f(k)); \]
so we can set $g(k) = kf(k)$.
Observe though that Pritchard's conjecture asks much less for large cuts than the thin tree conjecture.
Consider the strong forms.
For both conjectures, $T$ must have only $O(1)$ edges across every minimum cut, and across cuts close to minimum.
But for Pritchard's conjecture, $T$ may have more than $\eta k$ edges on cuts with at least $(1+\eta)k$ edges, i.e., thinness $1 -1/(\eta+1)$ suffices on these cuts.

Our result shows that we can find a tree that is thin enough on the ``most constrained'' cuts, and so we view it as progress towards Pritchard removable spanning tree conjecture.
A tree with thinness $\eta/(\eta+1)$ for cuts of size $(1+\eta)k$, for all $\eta$ above some constant, would miss at least $k$ edges from every cut $S$ with $|\delta(S)| \geq (1+\eta)k$,
more than what Pritchard's conjecture needs for these ``less constrained'' cuts. Notice that this requires a thinness of only $1-\frac{1}{k}$, for example, for cuts of size $k^2$. 
The existence of such a spanning tree is a very interesting open problem;
a resolution of this, together with this work, would likely lead to a proof of Pritchard's conjecture.
\subsection{Our Approach}

To prove \cref{thm:main-result}, we will reduce this problem to the strong thin tree conjecture for laminar families of cuts, which was shown to be true by Klein and Olver \cite{KO23}. 
Their result is phrased in the language of finding spanning trees satisfying prescribed cut bounds on a given laminar family. 
Here we state a version of their theorem for our setting.
\begin{restatable}[Implied by \cite{KO23}]{theorem}{thmthin}\label{thm:strong-tree-laminar}
Let $G=(V,E)$ be a $k$-edge-connected graph and $\cL$ be a laminar family of cuts. Then, there is a spanning tree $T$ with 
$|T \cap \delta(S)| \le \frac{66}{k}|\delta(S)|$ for all $S \in \mathcal{L}$. 

If in addition $|\delta(S)| \leq \tfrac43k$ for all $S \in \cL$, this can be improved to $|T \cap \delta(S)| \le 11$ for all $S \in \mathcal{L}$. 
Moreover, this spanning tree can be found in polynomial time. 
\end{restatable}
This is a slightly different form of the theorem stated in \cite{KO23}, and the quantitative improvement for our setting where all sets in the laminar family are near-minimum cuts requires unpacking their proof slightly. We discuss these details in \Cref{appendix}.

Given the existence of thin trees for laminar families, it is very natural to try and leverage this.
A plausible first attempt is to take an arbitrary maximal laminar family $\cL$ of near minimum cuts and find a tree that is thin with respect to $\cL$ using \cref{thm:strong-tree-laminar}.  In the case of minimum cuts, this works: taking a maximal laminar family of min cuts and finding a tree which is thin for $\cL$ implies that $|T \cap \delta(S)| \le O(1)$ for all minimum cuts $S$. This can be shown in a straightforward manner using the cactus representation of the minimum cuts \cite{DKL76}.
It is perhaps unsurprising that this works given that minimum cuts can be uncrossed; if $S$ and $T$ are crossing minimum cuts, then so are $S \cap T$ and $S \cup T$.

So, one might hope that this simple approach extends to the set of $\eta$-near minimum cuts for some $\eta > 0$. 
Unfortunately, this is not the case. For an extreme example, consider the following graph with an even number of vertices $n$ displayed in \cref{fig:badexample}. Begin with a Hamiltonian cycle $H_1=v_1,\dots,v_n$ in clockwise order with $k$ edges between adjacent vertices. Then, add a single copy of a second Hamiltonian cycle $H_2 = v_1,v_2,v_n,v_3,v_{n-1},v_4,v_{n-2},\dots,v_{n/2+1},v_1$. In other words, to construct $H_2$, begin with $v_1$ and then alternately add the first new vertex clockwise of $v_1$ and the first new vertex counterclockwise of $v_1$, terminating at $v_{n/2+1}$, at which point we return to $v_1$.  

\begin{figure}[htb!]
\centering 
\begin{tikzpicture}[scale=2, every node/.style={circle, draw, inner sep=1pt}] 
\def\n{8}
\def\angle{360/\n}
\foreach \i in {1,...,\n} {
    \node (v\i) at (135-\i*\angle:1) {$v_{\i}$};
}
\foreach \i in {1,...,\n} {
    \pgfmathtruncatemacro{\nexti}{mod(\i, \n) + 1}
    \foreach \j in {1,...,6} {
        \path[bend left=-30+10*\j-5] (v\i) edge (v\nexti);
    }
}
\path[thick, red,bend left=-40] (v1) edge (v2);
\path[thick, red,bend left=20] (v2) edge (v8);
\path[thick, red,bend left=10] (v3) edge (v8);
\path[thick, red,bend left=0] (v3) edge (v7);
\path[thick, red,bend left=-10] (v4) edge (v7);
\path[thick, red,bend left=-20] (v4) edge (v6);
\path[thick, red,bend left=-40] (v5) edge (v6);
\path[thick, red, dashed] (v1) edge (v5);

%

\path[draw, blue, thick] 
    ($(v3)!0.5!(v4) + (-0.37, 0.17)$) 
    ellipse [x radius=0.65cm, y radius=1.2cm, rotate=-30];

\end{tikzpicture}
\caption{A bad example for taking a maximal laminar family of near minimum cuts. In solid red is a tree which is thin with respect to the family $\cL$ described in this section ($\{\{v_1,v_2\},\{v_1,v_2,v_8\},\{v_1,v_2,v_3,v_8\}, \{v_1,v_2,v_3,v_7,v_8\},\dots \}$) but not thin with respect to the blue cut with interval $\{v_2,\dots,v_5\}$.}\label{fig:badexample}
\end{figure}
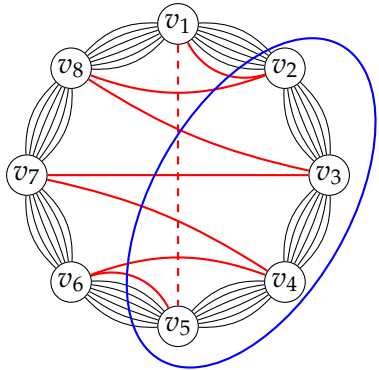

Taking $k \gg n$ (although this is not strictly necessary for the example to work), it is the case that \textit{every} interval of $H_1$ is a near minimum cut if we consider the set of $\eta$-near minimum cuts for any constant $\eta > 0$. Thus, we can choose $\cL$ to be any set of intervals. Define
$$\cL = \bigl\{\{v_1,v_2\},\{v_n,v_1,v_2\},\{v_n,v_1,v_2,v_3\}, \{v_{n-1},v_n,v_1,v_2,v_3\},\dots\bigr\}.$$ 
In other words, we choose $n-2$ sets to be in $\cL$, with the $i$th set consisting of the first $i+1$ vertices of $H_2$ (starting with $v_1$). But now, taking the edges of $H_2$ (without the last one) is a thin tree for $\cL$. But it is far from thin for the cut $\{v_2,\dots,v_{n/2}\}$; in fact, it has every edge of the tree.

We remark that the set of $\eta$-near minimum cuts \emph{cannot} be uncrossed; if $A$ and $B$ are crossing $\eta$-near minimum cuts, $A \cap B$ and $A \cup B$ are both necessarily $2\eta$-near minimum cuts, but this enlargement of $\eta$ can compound over taking combinations of many cuts.
This is the key reason that near minimum cuts can give rise to much richer structures than minimum cuts.

Despite this, we are able to show that a very carefully constructed laminar family suffices.  The construction is quite delicate. 
We make heavy use of the polygon representation of $\onefifth$-near-minimum cuts~\cite{Ben95, Ben97, BG08}, which we will discuss in detail in the next section. 

\nnote{Removed the kind of outline paragraph, which also talked about the outside atom stuff implying the inside atoms. I feel like this is hard to convey at this point, so I thought to move this to Section 3.}

\medskip

We remark that we have prioritized a more structured, modular analysis over obtaining the best possible constants.
We anticipate that a more refined analysis (possibly with some adjustments to the algorithm) can show the result for some larger $\eta$.
The polygon representation applies to the set of $\eta$-near minimum cuts for any $\eta \leq \onefifth$.
Thus, a natural question is whether our techniques can be made to work up to this limit; 
this seems to require some additional ideas.

Another natural question is whether the use of the polygon representation can be avoided in order to obtain the result for $\eta > \onefifth$. A result of this type would be of significant interest, as it could shed light on how to handle much more general cut families. 

\subsection{Related Work}


The thin tree conjecture is known to hold for planar and bounded genus graphs~\cite{OS11}. 
As already mentioned, the best known result for general graphs is $O(\poly\log\log n / k)$ thinness~\cite{AO15}; this result is not constructive. 
The best known constructive result gives thinness of $O(\frac{\log n}{\log \log n \cdot k})$~\cite{AGMOS17}.\footnote{We remark that since this is a randomized algorithm and it is not known how to test the thinness of a tree in polynomial time, this is not fully constructive.} 

The study of thinness for particular cut sets began with F{\"u}rer and Raghavachari \cite{FR92}, who constructed thin trees with respect to the singleton cuts. This predates the thin tree conjecture, but questions concerning the existence of spanning trees satisfying bounds across given cuts, in this case singleton cuts, are well-studied in their own right.
Goemans~\cite{Goe06} and Singh and Lau~\cite{SL15} studied minimum cost bounded degree spanning trees. 
Bansal et al.~\cite{BKKNP13} studied the problem of designing a spanning tree which crosses each cut in a laminar family at most a specified number of times, and showed how to obtain a tree with \emph{additive} $O(\log n)$ violation whenever the fractional relaxation is feasible.
Olver and Zenklusen \cite{OZ13} showed one can achieve thinness for any chain of cuts $S_1 \subset S_2 \subset \cdots \subset S_k \subset V$, and Linhares and Swamy \cite{LS16} resolved the minimum cost version of this question up to a small approximation factor in cost.

Near minimum cuts have been studied in several contexts, for example, Karger and Stein~\cite{KS96} gave bounds on how many such cuts can exist, and Bencz{\'u}r~\cite{Ben95,Ben97} and Bencz{\'u}r and Goemans~\cite{BG08} studied their structure and defined the polygon representation of near minimum cuts, which we will rely upon heavily in this work. Karlin, Klein, and Oveis Gharan~\cite{KKO21b} further studied the polygon representation to improve the integrality gap of the subtour polytope and proved some properties we will make use of in this work. Bansal, Cheriyan, Grout, and Ibrahimpur~\cite{BCGI24} studied the problem of finding a minimum cover of all $\alpha$-near minimum cuts in a $k$-edge-connected graph (i.e. a minimum cardinality set of edges so that each near minimum cut includes an edge from the set), and give a 16-approximation for any $\alpha$. 

\section{Preliminaries}\label{sec:prelim}

\subsection{Near Minimum Cuts}

Given a (multi-)graph $G=(V,E)$ and $S \subset V$, we will use $\delta(S)$ to denote the set of edges with exactly one endpoint in $S$. Recall a graph is $k$-edge-connected if $|\delta(S)| \ge k$ for all $\emptyset \subsetneq S \subsetneq V$.
\begin{definition}[$\eta$-Near Minimum Cut]
Given a multi-graph $G=(V,E)$ with minimum cut size $k$, a cut $(S,\overline{S})$ with $\emptyset \neq S \subsetneq V$ is an $\eta$-near minimum cut if $|\delta(S)| < (1+\eta)k$, where $\eta > 0$. We will use $\eta$-NMC (or sometimes NMC, when $\eta$ is clear from context) as a shorthand.
\end{definition}

We will often identify a cut $(S, \bar{S})$ with (an arbitrary) one of its shores $S$ and $\bar{S}$, as in the following definition.

\begin{definition}[Crossing Cuts]
    We say two cuts $S$ and $T$ cross if $S \cap T, S\setminus T, T \setminus S$, and $V \setminus (S \cup T)$ are all non-empty sets. 
\end{definition}

The following is standard and we simply cite a place it is proved. 
\Nathan{Added this editorial statement, feel free to remove it if it doesn't seem helpfu.}
\nnote{I like itd, but thought to move iy into the discussion in Section 1.2, and so remove it here.}
\begin{lemma}[\cite{OSS11}]\label{lem:cutdecrement}
For $G=(V,E,x)$, let $A \subsetneq V$ be an $\eta_A$-NMC and $B \subsetneq V$ an $\eta_B$-NMC, with $A$ and $B$ crossing. Then
$A\cap B, A\cup B, A\setminus B$ and $B\setminus A$ are all $(\eta_A+\eta_B)$-NMCs.
\end{lemma}





\subsection{Polygon Representation}

\begin{definition}[Connected Components] 
    Construct the \textit{cross graph} as follows. Create a vertex corresponding to each $\eta$-near minimum cut (represented as an arbitrary shore), and connect the vertices corresponding to two near minimum cuts $S,T$ if $S$ and $T$ cross. We define a \emph{connected component} of near minimum cuts as a collection of near minimum cuts that correspond to a connected component in the cross graph.  
\end{definition}

\begin{definition}[Atoms and Containment]
Given a connected component $\cC$ of near minimum cuts, let $\{a_i\}_{i \ge 0}$ be the coarsest partition of the vertices of the original graph so that for every $S \in \cC$, we have $a_i \subseteq S$ or $a_i \cap S = \emptyset$ for every $i$. We use $\cA(\cC)$ to represent $\{a_i\}_{i \ge 0}$, and call this the set of atoms. 
\end{definition}


Bencz{\'u}r and Goemans \cite{BG08} showed that any connected component of $\eta$-near minimum cuts admits a so-called \textit{polygon representation} so long as $\eta \leq \onefifth$. Here we define a polygon $P$ representing a connected component $\cC$ of $\eta$-near minimum cuts (for $\eta \leq \onefifth$) of a $k$-edge-connected graph.

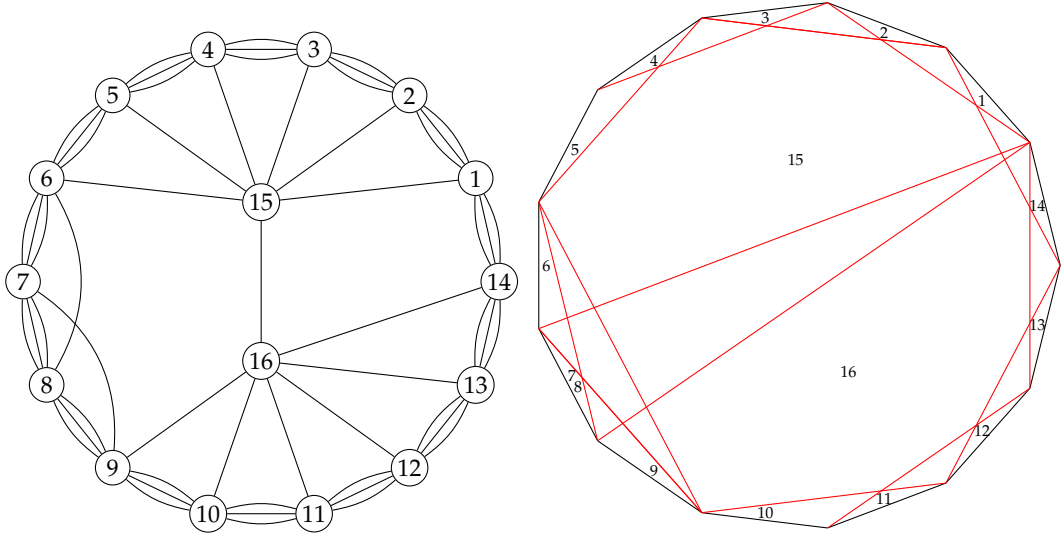
\begin{figure}[htb!]
\begin{center}
\def\deg{25.714285}
\def\deg{25.714285}
\begin{tikzpicture}[inner sep=1pt,minimum size=13pt, scale=.7,pre/.style={<-,shorten <=2pt,>=stealth,thick}, post/.style={->,shorten >=1pt,>=stealth,thick},scale=1.5]
\tikzstyle{every node} = [draw, circle,color=black];
\foreach \i in {1,...,14}{
\path (\i*\deg:3) node  (a_\i) {\footnotesize \i};
}
\path (0,1) node (c) {\footnotesize 15};
\path (0,-1) node (d) {\footnotesize 16};

\foreach \i [evaluate=\i as \j using {int(Mod(\i, 14)+1)}] in {1,...,14}{
	\path (a_\i) edge  (a_\j) edge [bend left=15] (a_\j) edge [bend right=15] (a_\j);
}
\foreach \i in {1,...,6}{
	\path (a_\i) edge (c);
}
\foreach \i in {9,...,14}{
	\path (a_\i) edge (d);
}
\path (c) edge (d);
\path (a_6) edge[bend left=30] (a_8);
\path (a_7) edge[bend left=30] (a_9);
\end{tikzpicture}
\begin{tikzpicture}[inner sep=1pt,minimum size=15pt, scale=.7,pre/.style={<-,shorten <=2pt,>=stealth,thick}, post/.style={->,shorten >=1pt,>=stealth,thick}]
\def\w{5}
\def\deg{27.692307}
\foreach \i in {1,...,6}{
\draw (\deg+\i*\deg:\w) -- (\deg+\deg+\i*\deg:\w);
\path (\i*\deg+\deg/2:\w-0.29) node  (a_\i) {\tiny \i};
}

\draw (\deg+7*\deg:\w) -- (\deg+\deg+7*\deg:\w);
\path (7*\deg-1.5+\deg/2:\w-0.29) node  (a_7) {\tiny 7};
\path (7*\deg+1.5+\deg/2:\w-0.29) node  (a_8) {\tiny 8};

\foreach \i [evaluate=\i as \j using {int(Mod(\i, 14)+1)}] in {8,...,13}{
\draw (\deg+\i*\deg:\w) -- (\deg+\deg+\i*\deg:\w);
\path (\i*\deg+\deg/2:\w-0.29) node  (a_\i) {\tiny \j};
}

\foreach \i in {1,...,3,6,8,9,10,11,12,13,14}{
\draw [color=red] (\deg+\i*\deg:\w) -- (\deg+\deg*2+\i*\deg:\w);
}
\path (0,2) node (c) {\tiny 15};
\path (1,-2) node (d) {\tiny 16};

\draw [color=red] (6*\deg:\w) -- (8*\deg:\w);
\draw [color=red] (7*\deg:\w) -- (9*\deg:\w);
\draw [color=red] (6*\deg:\w) -- (9*\deg:\w);
\draw [color=red] (7*\deg:\w) -- (14*\deg:\w);
\draw [color=red] (8*\deg:\w) -- (14*\deg:\w);
\end{tikzpicture}\end{center}
\caption{
Consider the graph on the left, with minimum cut 7. On the right is the polygon representation of the connected component of all proper cuts with at most 8 edges. This component consists of all proper near minimum cuts of the graph excluding the cut $\{7,8\}$, which is in its own connected component of size 1. As defined above, 15 and 16 are inside atoms, the others are outside atoms. Note $\{7,8\}$ is a single atom.}
\label{fig:polygonrepresentation}
\end{figure}
\begin{enumerate}
	\item A polygon representation is a convex regular polygon with a collection of \textit{diagonals} connecting vertices of the polygon. All polygon edges and diagonals are drawn using straight lines in the plane. The diagonals partition the polygon into \textit{cells}. 
	\item Each atom $a \in \cA(\cC)$ is mapped to a cell of the polygon.  If one of these cells is bounded by some portion of the polygon boundary it is {\em non-empty} and we call its atom an \textit{outside atom}. We call the atoms of all other non-empty cells \textit{inside atoms}. Note that some cells may not contain any atom. WLOG label the outside atoms $a_0,\dots,a_{m-1}$ in counterclockwise order, and label the inside atoms arbitrarily. We also label points of the polygon $p_0,\dots,p_{m-1}$ such that outside atom $a_i$ is on the side $(p_i,p_{i+1})$ and $a_0$ is on the side $(p_{m-1},p_0)$.
	\item No cell has more than one incident outer polygon edge.
	\item Each diagonal (often called a representing diagonal) defines a cut such that each side of the cut is given by the union of the atoms on each side. Furthermore, the collection of cuts given by these diagonals is exactly $\mathcal{C}$. 
\end{enumerate}

For set of vertices $S$ and a polygon $P$ for connected component $\cC$, let $O_P(S)$ denote the set of outside atoms contained in $S$. By the definition of polygons,  $S,S' \in \cC$ cross if and only if $O_P(S)$ and $O_P(S')$ cross. When the polygon is clear from context, we will sometimes simply use $O(S)$.

\subsection{Inside Atoms}

Inside atoms only exist if certain structures called $k$\textit{-cycles} appear in the connected component $\cC$ being represented by $P$.
\begin{definition}[{\cite[Definition 3]{BG08}}]\label{def:kcycle}
A family of sets $C_1,\dots,C_k\subseteq V$, for some $k \ge 3$, forms a $k$-cycle if
\begin{itemize}
\item $C_i$ crosses both $C_{i-1}$ and $C_{i+1}$ (we treat $C_{k+1}$ as $C_1$ and $C_{0}$ as $C_k$);
\item $C_i\cap C_j=\emptyset$ for $j\neq i-1, i$ or $i+1$; and
\item $\bigcup_{1\leq i\leq k} C_i \neq V$.
\item If $k=3$, we have the additional condition $(C_i \cap C_{i+1}) \not\subseteq C_{i-1}$ for $i \in \{1,2,3\}$.
\end{itemize}
\end{definition}

\begin{lemma}[{\cite[Lemma 22]{BG08}}]\label{lem:no-kcycle}
Any $k$-cycle formed by cuts in a connected component ${\cal C}$ of $\eta$-near min cuts satisfies $k\geq 1/\eta$. (Note if $\eta = 0$ then there are no $k$-cycles.)
\end{lemma}

\begin{lemma}[{\cite[Definition 4]{BG08}}]\label{lem:inside-implies-cycle}
An atom $a \in {\cal A}({\cal C})$ is an inside atom if and only if there is a $k$-cycle $C_1,\dots,C_k \in {\cal C}$ such that $a \not\in C_i$ for all $1 \le i \le k$. 
\end{lemma}

In \cite{KKO21b}, the following lemmas were shown. The following lemma shows that cells of $P$ that can be described with few diagonals and do not intersect any side of $P$ are empty. 

\begin{lemma}[{\cite[Lemma 4.30]{KKO21b}}]\label{thm:halfplanes}
Let $\eta \leq  \onefifth$ and let $P$ be the polygon representation for a connected component $|{\cal C}|>1$ of $\eta$-NMCs of a graph. Suppose $H$ is the intersection of half-planes $H_0,\dots,H_{\ell-1}$ corresponding to diagonals $D_0,\dots,D_{\ell-1}$ of $P$ that has a positive area. If $H$ does not contain any side of $P$ (equivalently, it does not contain any outside atom) and $\ell<1/(2\eta)$ then $H$ does not have any inside atoms.
\end{lemma}

It is clear that in any polygon representing a connected component of cuts $\cC$, for any $S,T \in \cC$, $S = T$ if and only if $O(S) = O(T)$ and similarly $S,T$ cross if and only if $O(S),O(T)$ cross. The following two lemmas extend this notion to all $\onefifth$-NMC, even in polygons which were constructed for smaller values of $\eta$. 

\begin{lemma}[{\cite[Lemma 4.23]{KKO21b}}]\label{lem:extendedprop20}
Let $P$ be the polygon representation for a connected component $|{\cal C}|>1$ of $\eta$-NMCs of a graph $G$ with atom set $\cA(\cC)$. If $A,B\subsetneq \cA(\cC)$ are two $\onefifth$-NMCs with $O_P(A) = O_P(B)\ne \emptyset$ and there is an atom $r\in \cA(\cC)$ such that $r\notin A,B$, then $A=B$.
\end{lemma}

\begin{lemma}[{\cite[Lemma 4.24]{KKO21b}}]\label{lem:canonicalcrossing}
Let $P$ be the polygon representation for a connected component $|{\cal C}|>1$ of $\eta$-NMCs of a graph $G$ with atom set $\cA(\cC)$. If $A,B\subsetneq \cA(\cC)$ are two $\onefifth$-NMCs with $O_P(A) = O_P(B)\ne \emptyset$. Then, $A$ and $B$ cross if and only if $O(A)$ and $O(B)$ cross.
\end{lemma}

\section{Constructing the Laminar Family}\label{sec:outside_atom_case}
Say that a collection of cuts $\Cover$ \emph{covers} an $\eta$-NMC $S$ if $\delta(S) \subseteq \cup_{B \in \Cover} \delta(B)$. 
For $\eta = \onefortieth$, we will show how to construct a laminar family $\cL$ of $4\eta$-NMCs such that every $\eta$-NMC $S$ is covered by a set $\Cover(S) \subseteq \cL$ of constant size.
This means that any spanning tree $T$ and $\eta$-NMC $S$, 
$$|\delta(S) \cap T| \leq \sum_{B \in \Cover(S)} |\delta(B) \cap T| \le 11 \cdot |B| \le O(1),$$
where the second inequality follows from \Cref{thm:strong-tree-laminar} (which we may apply as $4 \eta < \onefifth$). Hence, the existence of such $\cL$ implies \Cref{thm:main-result}. 

We will prove a version of this coverage statement for each connected component of $\eta$-NMCs, with the help of the polygon representation.
In \Cref{subsec:full-laminar-family} we will see that the laminar families for each connected component can be easily combined to construct the overall laminar family.

The coverage theorem for a single nontrivial connected component is the following.
\begin{theorem}\label{thm:L-covers-nmc-general}
    Let $\eta = \onefortieth$. Given a connected component $\cC$ of $\eta$-NMC with $|\cC| > 1$, there exists a laminar family $\cL$ such that every $\eta$-NMC in $\cC$ is covered by a collection of at most $8$ cuts in $\cL$. This collection $\cL$ can be computed in polynomial time given $\cC$, and for each atom $a \in \cA(\cC)$ and every cut $S \in \cL$, we have $a \subseteq S$ or $a \cap S = \emptyset$.  
\end{theorem}

In \Cref{subsec:algorithm}, we describe our algorithm for constructing the laminar family.
In \Cref{subsec:outside-atoms}, we show a weaker coverage statement that involves only outside atoms.
Then in \Cref{subsec:general-case} we use  properties of the polygon representation described in \cref{sec:prelim} to show that coverage holds fully.
We find it slightly surprising that for the bulk of our analysis, we can eschew the complexity of the full polygon representation and focus only on the outside atoms, which have a simpler structure. 
\nnote{Not sure if this move was fully successful.}



For the remainder, fix $\eta = \onefortieth$, a connected component $\cC$ of $\eta$-NMCs with $|\cC| > 1$, and its polygon representation $P$. 
We let $\cA$ denote the atoms of $P$.
Let the vertices of $P$ be $p_0,p_1,\dots,p_{n-1}$, and the outside atoms of $P$ be $a_0,a_1,\dots,a_{n-1}$, in counterclockwise order such that $a_i$ is the only outside atom in one side of the cut with representing diagonal $(p_i,p_{i+1})$, where indices are taken modulo $n$. 
We will refer to $a_0$ as the root atom; $a_n$ will be another name for $a_0$.

\begin{definition}[Intervals]
Let $l,r$ be such that $1 \le l < r \le n$. 
We define $\out{l}{r}$ to be the set of outside atoms $\{ a_l, a_{l+1}, \ldots, a_{r-1}\}$; we will refer to such a set as an \emph{interval}. 
\end{definition}
Note that $\out{l}{r}$ includes $a_l$ but not $a_r$; it can be thought of as the set of outside atoms between $p_l$ and $p_r$ on the polygon.

\nnote{Tweaked the following.}
By \Cref{lem:extendedprop20}, no two distinct $\onefifth$-NMCs have the same set of outside atoms.
Thus the following is well-defined.
%

\begin{definition}[Shadows and near minimum intervals]\label{def:canonical-cut}
    Call an interval $\out{l}{r}$ a \emph{shadow} if there exists a $\onefifth$-NMC $S$ with $O(S) = \out{l}{r}$, and in this case, define $\cut{l}{r} = S$. 
    We refer to $\cut{l}{r}$ as the \emph{canonical cut} associated with the interval $\out{l}{r}$.
   
    If in addition $\cut{l}{r}$ is an $\eta$-NMC, we call $\out{l}{r}$ a \emph{near minimum interval}, or NMI for short.
%
%
\end{definition}

The following is an immediate consequence of \Cref{lem:canonicalcrossing}.
\begin{corollary}\label{fact:canonicalcrossing}
If $\cut{l}{r}$ and $\cut{l'}{r'}$ are crossing canonical cuts, then $\out{l}{r}$ and $\out{l'}{r'}$ are crossing intervals.
\end{corollary}

\subsection{Algorithm}\label{subsec:algorithm}

In this subsection, we describe a recursive top-down procedure to construct $\cL$, the desired laminar family. Before describing the algorithm formally, we start by giving some intuition for the algorithm and an overview of its general steps. For this intuition, assume that there is a single polygon and there are no inside atoms.

If the set of near minimum intervals is already laminar, we are done. At the other extreme, if all intervals $\out{l}{r}$ are near minimum, then we can set $\cL = \{\cut{1}{2},\cut{1}{3},\dots,\cut{1}{n}\}$ because for all $l,r$, $\cut{l}{r}$ is covered by $\cut{1}{r}$ and $\cut{1}{l}$. In fact, $\cL = \{\cut{1}{2},\cut{1}{3},\dots,\cut{1}{n}\}$ would always work if $\out{ 1}{r}$ is the shadow of a $4\eta$-NMC for all $r$. Unfortunately, this is not always the case.

Instead, suppose we greedily add all cuts of the form $\cut{1}{r}$ that are $4\eta$-NMCs to $\cL$. Let $S = \cut{l_S}{r_S}$ be an $\eta$-NMC. Let 
$$i = \min\{x: l_S \le x \le r_S, \langle 1,x\rangle\text{ is a shadow of a }4\eta\text{-NMC}\}$$
$$j = \max\{x: l_S \le x \le r_S, \langle 1,x\rangle\text{ is a shadow of a }4\eta\text{-NMC}\}.$$
Suppose $i$ and $j$ exists. Then, $\cut{1}{i}$ and $\cut{1}{j}$ covers all edges of $\delta(S)$ except the edges between $\out{1}{l_S}$ and $\out{l_S}{i}$ and the edges between $\out{j}{r_S}$ and $\out{r_S}{n}$. If we could add $\cut{l_S}{i}$ and $\cut{j}{r_S}$ (assuming they exist and are $4\eta$-NMCs), then we can cover all edges of $\delta(S)$. We can almost do that: since $\cut{1}{i}$ is a $4\eta$-NMC and $\cut{l_S}{r_S}$ is an $\eta$-NMC, by \Cref{lem:cutdecrement}, $\cut{l_S}{i}$ exists and is a $5\eta$-NMC. In the actual algorithm, instead of greedily taking all prefix $4\eta$-NMCs, we will only take a subset of $2\eta$-NMCs we define as special cuts (see \Cref{def:special-cut}). We argue that in that case, $\cut{l_S}{i}$ and $\cut{j}{r_S}$ are also special.

It is still unclear whether we can add both $\cut{l_S}{i}$ and $\cut{j}{r_S}$ to $\cL$. A natural way to proceed with the algorithm is to recurse on all subintervals $\langle x_k,x_{k+1}\rangle$ where $x_1<x_2<\dots<x_m$ are indices such that $\out{1}{x_k}$ is special. Suppose $x_a=i$ and $x_b=j$. When we recurse on $\langle x_{a-1},x_a\rangle$, we need to add special intervals from the right to add $\cut{l_S}{i}$, but when we recurse on $\langle x_{b},x_{b+1}\rangle$, we need to add special intervals from the left to add $\cut{j}{r_S}$. How does the algorithm know which direction to add intervals? A wishful attempt is to alternate between adding intervals from the left and right, and surprisingly this works (modulo some technical additions). It turns out that for any $S = \cut{l_S}{r_S}$, we can pick $O(1)$ sets added to $\cL$ in $3$ recursion layers that covers $S$.



We define special cuts and intervals that are alluded in the outline.
\begin{definition}[Special cuts and intervals]\label{def:special-cut}
    We say that a canonical cut $S=\cut{l}{r}$ is \emph{special} if either
    \begin{enumerate}[(i)]
        \item $S$ is an $\eta$-NMC, or 
        \item there exist two crossing canonical cuts $A$ and $B$, both of which are $\eta$-NMCs, such that $S$ is the intersection, union or set difference of $A$ and $B$.
    \end{enumerate}
    

    We say that an interval $\out{l}{r}$ is \emph{special} if it is a shadow and $\cut{l}{r}$ is a special cut.
\end{definition}
By \Cref{lem:cutdecrement}, any special cut is a $2\eta$-NMC; however, not every $2\eta$-NMC is necessarily special. Similarly, the set difference of two crossing special cuts is a $4\eta$-NMC, since $4\eta < \onefifth$.
Note that by \Cref{fact:canonicalcrossing}, a special interval is the intersection, union or set difference of crossing NMIs.
\def\scale{0.7}
\begin{figure}[t!]\centering
\begin{tikzpicture}[scale=\scale]
		\draw [color=black] (-2,0) -- (16,0);
        \node [color=black,inner sep=2,label={[yshift=-0.9cm]$l_{A}$},circle,fill] at (0,0) () {};
        \node [color=black,inner sep=2,label={[yshift=-0.9cm]$l_{B}$},circle,fill] at (4,0) () {};
        \node [color=black,inner sep=2,label={[yshift=-0.9cm]$r_A$},circle,fill] at (8,0) () {};
        \node [color=black,inner sep=2,label={[yshift=-0.9cm]$r_{B}$},circle,fill] at (12,0) () {};
        \node [color=black,inner sep=2,label={[yshift=-0.9cm]$x$},circle,fill] at (-2,0) () {};
        \draw [color=red!20,line width=2pt] (-2,0.25) -- (8,0.25);
        \draw [color=red!20,line width=2pt] (4,0.5) -- (12,0.5);
\end{tikzpicture}

\begin{flushleft}
\emph{Case 2.} The red intervals show the NMIs used to prove that $A \cap B$ and $B \setminus A$ are special.
\end{flushleft}

\begin{tikzpicture}[scale=\scale]
		\draw [color=black] (-2,0) -- (16,0);
        \node [color=black,inner sep=2,label={[yshift=-0.9cm]$l_{A}$},circle,fill] at (2,0) () {};
        \node [color=black,inner sep=2,label={[yshift=-0.9cm]$l_{B}$},circle,fill] at (4,0) () {};
        \node [color=black,inner sep=2,label={[yshift=-0.9cm]$r_A$},circle,fill] at (8,0) () {};
        \node [color=black,inner sep=2,label={[yshift=-0.9cm]$r_{B}$},circle,fill] at (12,0) () {};
        \node [color=black,inner sep=2,label={[yshift=-0.9cm]$y$},circle,fill] at (0,0) () {};
        \node [color=black,inner sep=2,label={[yshift=-0.9cm]$x$},circle,fill] at (-2,0) () {};
        \draw [color=red!20,line width=2pt] (0,0.25) -- (8,0.25);
        \draw [color=red!20,line width=2pt] (4,0.5) -- (12,0.5);
\end{tikzpicture}

\begin{flushleft}
    \emph{Case 3.} The red intervals show the NMIs used to prove that $A \cap B$ and $B \setminus A$ are special.
\end{flushleft}

\begin{tikzpicture}[scale=\scale]
		\draw [color=black] (-2,0) -- (16,0);
        \node [color=black,inner sep=2,label={[yshift=-0.9cm]$l_{A}$},circle,fill] at (2,0) () {};
        \node [color=black,inner sep=2,label={[yshift=-0.9cm]$l_{B}$},circle,fill] at (4,0) () {};
        \node [color=black,inner sep=2,label={[yshift=-0.9cm]$r_A$},circle,fill] at (8,0) () {};
        \node [color=black,inner sep=2,label={[yshift=-0.9cm]$r_{B}$},circle,fill] at (12,0) () {};
        \node [color=black,inner sep=2,label={[yshift=-0.9cm]$x_2$},circle,fill] at (14,0) () {};
        \node [color=black,inner sep=2,label={[yshift=-0.9cm]$y_2$},circle,fill] at (16,0) () {};
        \node [color=black,inner sep=2,label={[yshift=-0.9cm]$x_1$},circle,fill] at (10,0) () {};
        \draw [color=red!20,line width=2pt] (2,-0.25) -- (10,-0.25);
        \draw [color=red!20,line width=2pt] (4,-0.5) -- (12,-0.5);
        \draw [color=blue!20,line width=2pt] (8,0.25) -- (16,0.25);
        \draw [color=blue!20,line width=2pt] (2,0.5) -- (14,0.5);
        \draw [color=blue!20,line width=2pt] (4,0.75) -- (12,0.75);
\end{tikzpicture}

\begin{flushleft}
\emph{Case 4.} $x_1$ and $x_2$ represent possible positions of $x$ relative to $r_B$. The red intervals show the NMIs used to prove that $A \setminus B$ and $A \cup B$ are special if $x<r_B$. The blue intervals show the NMIs used to prove that $A \cap B$ and $B \setminus A$ are special if $x \ge r_B$.
\end{flushleft}

\begin{tikzpicture}[scale=\scale]
		\draw [color=black] (-2,0) -- (16,0);
        \node [color=black,inner sep=2,label={[yshift=-0.9cm]$l_{A}$},circle,fill] at (-2,0) () {};
        \node [color=black,inner sep=2,label={[yshift=-0.9cm]$l_{B}$},circle,fill] at (4,0) () {};
        \node [color=black,inner sep=2,label={[yshift=-0.9cm]$r_A$},circle,fill] at (10,0) () {};
        \node [color=black,inner sep=2,label={[yshift=-0.9cm]$r_{B}$},circle,fill] at (16,0) () {};
        \node [color=black,inner sep=2,label={[yshift=-0.9cm]$x_2$},circle,fill] at (7,0) () {};
        \node [color=black,inner sep=2,label={[yshift=-0.9cm]$y_1$},circle,fill] at (0,0) () {};
        \node [color=black,inner sep=2,label={[yshift=-0.9cm]$x_1$},circle,fill] at (2,0) () {};
        \draw [color=red!20,line width=2pt] (0,-0.25) -- (10,-0.25);
        \draw [color=red!20,line width=2pt] (4,-0.5) -- (16,-0.5);
        \draw [color=blue!20,line width=2pt] (-2,0.25) -- (7,0.25);
        \draw [color=blue!20,line width=2pt] (4,0.5) -- (16,0.5);
\end{tikzpicture}

\begin{flushleft}
    \emph{Case 5.} $x_1$ and $x_2$ represent possible positions of $x$ relative to $l_B$. The red intervals show the NMIs used to prove that $A \cap B$ and $B \setminus A$ are special. The blue intervals show the NMIs used to prove that $A \setminus B$ and $A \cup B$ are special.
\end{flushleft}

\caption{Figures depicting the proof of \Cref{lem:special-cross-nmc}.}
\label{fig:special-cross-nmi}
\end{figure}

Here is a useful lemma about special intervals that will be used repeatedly in our proofs. \Nathan{Let's try to prove this using 3-way submodularity.}
\begin{lemma}\label{lem:special-cross-nmc}
Let $A = \out{l_A}{r_A}$ be a special interval and $B = \out{l_B}{r_B}$ be an NMI crossing $A$. Then, 
\begin{itemize}
\item at least one of $A \cap B$ and $A \setminus B$ is special.
\item at least one of $A \cup B$ and $B \setminus A$ is special.
\end{itemize}
\end{lemma}
\begin{proof}
    Without loss of generality, assume that $l_A<l_B<r_A<r_B$. There are several cases depending on how $A$ is special; we refer to \Cref{fig:special-cross-nmi} for an illustration of all cases other than the first.
\begin{enumerate}
\item $A$ is an NMI.

Then, $A \cap B$, $A \cup B$, $A \setminus B$, and $B \setminus A$ are all special by definition.
\item $A = \out{x}{r_A} \cap \out{l_A}{y}$ for some NMIs $\out{x}{r_A}$ and $\out{l_A}{y}$ with $x < l_A < r_A < y$.

Then, $A \cap B = \out{l_B}{r_A} = \out{x}{r_A} \cap B$ and $B \setminus A = \out{r_A}{r_B} = B \setminus \out{x}{r_A}$ are special.
\item $A = \out{y}{r_A} \setminus \out{x}{l_A}$ for some NMIs $\out{y}{r_A}$ and $\out{x}{l_A}$ with $x<y<l_A<r_A$.

Then, $A \cap B = \out{l_B}{r_A} = \out{y}{r_A} \cap B$ and $B \setminus A = \out{r_A}{r_B} = B \setminus \out{y}{r_A}$ are special. 
\item $A = \out{l_A}{x} \setminus \out{r_A}{y}$ for some NMIs $\out{l_A}{x}$ and $\out{r_A}{y}$ with $l_A < r_A < x < y$.

If $x<r_B$, then $A \setminus B = \out{l_A}{l_B} = \out{l_A}{x} \setminus B$ and $A \cup B = \out{l_A}{r_B} = \out{l_A}{x} \cup B$ are special. Otherwise, $x \ge r_B$, so $y>r_B$. Then, $A \cap B = \out{l_B}{r_A} = B \setminus \out{r_A}{y}$ and $B \setminus A = \out{r_A}{r_B} = \out{r_A}{y} \cap B$ are special. 
\item $A = \out{l_A}{x} \cup \out{y}{r_A}$ for some NMIs $\out{l_A}{x}$ and $\out{y}{r_A}$ with $l_A<y<x<r_A$.

If $x \le l_B$, then $y < l_B$, so $A \cap B = \out{l_B}{r_A} = \out{y}{r_A} \cap B$ and $B \setminus A = \out{r_A}{r_B} = B \setminus \out{y}{r_A}$ are special. Otherwise, $x > l_B$, so $A \setminus B = \out{l_A}{l_B} = \out{l_A}{x} \setminus B$ and $A \cup B = \out{l_A}{r_B} = \out{l_A}{x} \cup B$ are special. 
\qedhere 
\end{enumerate} 
\end{proof}

Our algorithm for constructing a laminar family of near minimum cuts is described in \Cref{alg:construct-laminar}.
The procedure is recursive; our family $\mathcal{L}$ is generated by running $\Cons(1,n,0,0)$ on the outside atoms of our polygon $P$.
Each call either proceeds ``from the left'' or ``from the right''; 
crucially, this direction alternates as one goes down the recursion tree.
Suppose we are proceeding from the left on some interval $\out{L}{R}$ (everything is symmetric). 
The algorithm finds a maximal chain $\out{L}{x_1} \subset \out{L}{x_2} \subset  \ldots \subset \out{L}{x_m}$ of special sets within the interval.
The algorithm recurses into the rings $\out{x_i}{x_{i+1}}$ of this chain, including $\out{L}{x_1}$ and $\out{x_m}{R}$.
Except for possibly $\out{x_m}{R}$, all these intervals are special or set differences of special intervals, and thus their corresponding cuts $\cut{x_i}{x_{i+1}}$ are $4\eta$-minimum cuts. These are all added to $\cL$. 
The special cuts $\cut{L}{x_i}$ are also added. 

One technical addition is the $\untouched$ counter. 
It may be that the algorithm is applied to some interval $\out{L}{R}$, say from the left, but there are no special cuts of the form $\out{L}{x}$.
The algorithm will then recurse on the same interval, but now from the right.
If, however, there are also no special intervals of the form $\out{x}{R}$, we  shrink the interval to make sure the algorithm continues.

\begin{algorithm}[tb]
\begin{algorithmic}[1]
    \Require{A tuple $(L,R,{\depth},\untouched)$.}
    \Ensure{A laminar family of $4\eta$-near minimum cuts.}
    \If{$L \ge R$}
    \State \Return{$\emptyset$.}
    \ElsIf{$\untouched \ge 2$}
    \State \Return{$\Cons(L+1,R,{\depth}+1,0)$.}
    \EndIf
    \State $\mathcal{L} := \emptyset$.
    \If{$\out{L}{R}$ is a special interval or the set difference of two crossing special intervals} 
    \State\label{line:add-active-interval} $\mathcal{L} := \{\cut{L}{R}\}$.
    \EndIf
    \If{${\depth}$ is even}
\State Let $x_1<x_2<\dots<x_m$ be the elements of $S = \{ x : L < x \leq R \text{ and } \out{L}{x}\text{ is special}\}$. Let $x_0 = L$ and $x_{m+1}=R$. 
        \If{$m \ge 1$}
            \State\label{line:prefix-chains} \Return{$\mathcal{L} \cup \bigcup_{i=1}^{m}\cut{L}{x_i} \cup \bigcup_{i=0}^{m}\Cons(x_i,x_{i+1},{\depth}+1,0)$}. 
        \Else
            \State \Return{$\Cons(L,R,{\depth}+1,\untouched+1)$.}
        \EndIf
    \Else
    \State Let $x_1<x_2<\dots<x_m$ be the elements of $S = \{x : L \leq x < R \text{ and } \out{x}{R}\text{ is special}\}$. Let $x_0 = L$ and $x_{m+1}=R$. 
        \If{$m \ge 1$}
            \State\label{line:suffix-chains} \Return{$\mathcal{L} \cup \bigcup_{i=1}^{m}\cut{x_i}{R} \cup \bigcup_{i=0}^{m}\Cons(x_i,x_{i+1},{\depth}+1,0)$}. 
        \Else
            \State \Return{$\Cons(L,R,{\depth}+1,\untouched+1)$.}
        \EndIf
    \EndIf
\end{algorithmic}
\caption{Procedure $\Cons(L, R, {\depth},\untouched)$.}\label{alg:construct-laminar}
\end{algorithm}

Let $\cL$ be the family obtained from $\Cons(1,n,0,0)$.
\begin{lemma}
$\cL$ is a laminar family of $4\eta$-near minimum cuts.
\end{lemma}
\begin{proof}
\nnote{I modified this proof; probably unnecessarily. Check if OK.}
    All the cuts added to $\cL$ are either special or the setwise difference of two crossing special cuts, and so are $4\eta$-near minimum cuts by \cref{lem:cutdecrement}. 
    
    To show that $\cL$ is a laminar family, firstly note that all the cuts added by \Cref{alg:construct-laminar} are of the form $\cut{l}{r}$ with $1 \le l < r \le n$. Hence, by \cref{fact:canonicalcrossing}\Nathan{added this ref to 3.4} $\cL$ is laminar if and only if the corresponding intervals $O(\cL) := \{ \out{l}{r} : \cut{l}{r} \in \cL\}$ form a laminar family.

We prove by induction on $R-L$ that $\Cons(L,R,\cdot,\cdot)$ returns a family of $\cL$ such that $O(\cL$) is a laminar family of intervals contained in $\out{L}{R}$, which suffices by the above observation. 

The base case $R-L=0$ is obvious.
For the inductive step, consider a call of the form $\Cons(L,R,{\depth},\untouched)$. Either the algorithm calls $\Cons$ with a subinterval of $\out{L}{R}$ at a higher depth, or line \ref{line:prefix-chains} or line \ref{line:suffix-chains} was triggered. If line \ref{line:prefix-chains} was triggered (the other case is analogous), then the cuts added to $\cL$ before recursing are $\cut{L}{x_1}$, $\cut{L}{x_2}, \ldots, \cut{L}{x_m}$, and possibly $\cut{L}{R}$, with corresponding intervals $\out{L}{x_1}$, $\out{L}{x_2}, \ldots, \out{L}{x_m}$ and $\out{L}{R}$ included in $O(\cL)$. Furthermore, the algorithm then recurses into each subinterval of the form $\out{x_i}{x_{i+1}}$, as well as $\out{L}{x_1}$ and $\out{x_m}{R}$. 
By the induction hypothesis, each of these adds a laminar family contained within some $\out{x_{i-1}}{x_i}$ to $O(\cL)$, and so $O(\cL)$ is laminar.
\end{proof}



\subsection{Analysis for Outside Atoms}\label{subsec:outside-atoms}

Here, we will prove a weaker coverage statement than is claimed by \Cref{thm:L-covers-nmc-general}, that only refers to outside atoms.

\begin{definition}
    Given a pair $(u,v)$ of atoms, we say that a cut $C$ \emph{covers} $(u,v)$ if $u$ and $v$ lie on different sides of $C$, and we say that a collection of cuts $\Cover$ \emph{covers} $(u,v)$ if at least one cut in $\Cover$ covers $(u,v)$.
%
\end{definition}

\begin{lemma}\label{lem:outside-edges-covered}
Let $S$ be an $\eta$-near minimum cut.
Then there exists a collection of cuts $\Cover \subseteq \cL$ of size at most $8$ such that every pair of outside atoms covered by $S$ is covered by $\Cover$. 
Furthermore, all of the sets in $\Cover$ are special except for at most one, which is the set difference of two crossing special cuts.
\end{lemma}
Note that if ``every pair of outside atoms'' was instead ``every pair of atoms'', this would yield \Cref{thm:L-covers-nmc-general}.
\begin{proof}
Since $S$ is an $\eta$-near minimum cut, $S = \cut{\ltarget}{\rtarget}$ for some $1 \le \ltarget < \rtarget \le n$. We say that an interval $\out{L}{R}$ is \emph{active at depth $d$} if $\Cons(L,R,d,\cdot)$ was called at some point. Let 
\[ \Cons(\lactive1,\ractive1,\depthact,\untouchedact) \]
be the call of $\Cons$ with the highest depth such that $\out{\ltarget}{\rtarget} \subseteq \out{\lactive1}{\ractive1}$. 

Firstly, we show that $\untouchedact < 2$. 
If this is not the case, then 
by maximality of $\depthact$ we must have that $\lactive1=\ltarget$, and $\out{\lactive1}{x}$ is not special for any $x$ such that $\lactive1<x\le \ractive1$. 
However, $\out{\ltarget}{\rtarget} = \out{\lactive1}{\rtarget}$ is an NMI, a contradiction.

Without loss of generality, assume that $\depthact$ is even. Since $\untouchedact < 2$, there exists some $x$ such that $\ltarget<x<\rtarget$ and $\out{\lactive1}{x}$ is special (or else $\out{\ltarget}{\rtarget}$ is contained in some active interval with depth $\depthact+1$). Let 
\begin{align*}
    &\lone = \min\{x: \ltarget \le x \le \rtarget, \out{\lactive1}{x}\text{ is special}\}\\ &\rone = \max\{ x: \ltarget \le x \le \rtarget,  \out{\lactive1}{x}\text{ is special}\}. 
\end{align*}
Since $\depthact$ is even and $\out{\lactive1}{\lone}$, $\out{\lactive1}{\rone}$ are special, $\cut{\lactive1}{\lone}$ and $\cut{\lactive1}{\rone}$ are added to $\cL$ by line \ref{line:prefix-chains} of \Cref{alg:construct-laminar}. 
These two cuts cover all pairs of outside atoms covered by $\out{\ltarget}{\rtarget}$, except for pairs involving either
\begin{enumerate}[(i)]
\item an outside atom in $\out{\lactive1}{\ltarget}$ and an outside atom in $\out{\ltarget}{\lone}$, or
\item an outside atom in $\out{\rone}{\rtarget}$ and an outside atom not in $\out{\lactive1}{\rtarget}$.
\end{enumerate}
\Cref{fig:construct-starting-config} shows the relative positions of the defined points.
\def\scale{0.7}
\begin{figure}[htb!]\centering
\begin{tikzpicture}[scale=\scale]
		\draw [color=black] (-2,0) -- (16,0);
        \node [color=red,inner sep=2,label={[yshift=-0.9cm]$\lactive1$},circle,fill] at (0,0) () {};
        \node [color=red,inner sep=2,label={[yshift=-0.9cm]$\ractive1$},circle,fill] at (14,0) () {};
        \node [color=blue,inner sep=2,label={[yshift=-0.9cm]$\ltarget$},circle,fill] at (2,0) () {};
        \node [color=blue,inner sep=2,label={[yshift=-0.9cm]$\rtarget$},circle,fill] at (11,0) () {};
        \node [color=black,inner sep=2,label={[yshift=-0.9cm]$\lone$},circle,fill] at (4,0) () {};
        \node [color=black,inner sep=2,label={[yshift=-0.9cm]$\rone$},circle,fill] at (8,0) () {};
        \node [color=black,inner sep=2,label={[yshift=-0.9cm]$1$},circle,fill] at (-2,0) () {};
        \node [color=black,inner sep=2,label={[yshift=-0.9cm]$n$},circle,fill] at (16,0) () {};
        \draw [color=red!20,line width=2pt] (0,0.25) -- (4,0.25);
        \draw [color=red!20,line width=2pt] (0,0.5) -- (8,0.5);
\end{tikzpicture}
\caption{\small  A figure showing the relative locations of the points defined in the first part of the proof of \Cref{lem:outside-edges-covered}.
The cut we are analyzing is $\out{\ltarget}{\rtarget}$. The cuts $\cut{\lactive1}{\lone}$ and $\cut{\lactive1}{\rone}$ are in $\cL$ and are included in $\Cover$.}
	\label{fig:construct-starting-config}
\end{figure}

Firstly, we prove that the pairs of outside atoms such that one is in $\out{\lactive1}{\ltarget}$ and the other is in $\out{\ltarget}{\lone}$ are covered by a cut in $\mathcal{L}$.
\begin{claim}\label{lem:left-edges}
There is a special cut in $\cL$ that covers every pair of outside atoms such that one is in $\out{\lactive1}{\ltarget}$ and the other is in $\out{\ltarget}{\lone}$.
\end{claim}
\begin{nestedproof}
If $\ltarget=\lone$, no such pairs exist and we are done. Otherwise, $\out{\lactive1}{\ltarget}$ is not special. 
We show that $\cut{\ltarget}{\lone}$ is active at depth $\depthact+1$ and is added to $\cL$; this will be our desired cut. 
Note that 
there will be some active interval of depth $\depthact + 1$ ending at $\lone$ containing $\ltarget$. Since $\depthact+1$ is odd, if $\out{\ltarget}{\lone}$ is special, then line \ref{line:suffix-chains} of \Cref{alg:construct-laminar} will add $\cut{\ltarget}{\lone}$ to $\mathcal{L}$. Thus, it is sufficient to show that $\out{\ltarget}{\lone}$ is special.

By \Cref{lem:special-cross-nmc} with $A = \out{\lactive1}{\lone}$ and $B = \out{\ltarget}{\rtarget}$, either $A \cap B$ or $A \setminus B$ is special. 
However, $A \setminus B = \out{\lactive1}{\ltarget}$ is not special from the definition of $\lone$, since $\ltarget < \lone$. 
Hence, $A \cap B = \out{\ltarget}{\lone}$ is special, as claimed. \Cref{fig:ltargetlone} shows the intervals used in the argument.
\def\scale{0.7}
\begin{figure}[htb!]\centering
\begin{tikzpicture}[scale=\scale]
		\draw [color=black] (-2,0) -- (16,0);
        \node [color=black,inner sep=2,label={[yshift=-0.9cm]$\lactive1$},circle,fill] at (0,0) () {};
        \node [color=black,inner sep=2,label={[yshift=-0.9cm]$\ractive1$},circle,fill] at (14,0) () {};
        \node [color=black,inner sep=2,label={[yshift=-0.9cm]$\ltarget$},circle,fill] at (2,0) () {};
        \node [color=black,inner sep=2,label={[yshift=-0.9cm]$\rtarget$},circle,fill] at (11,0) () {};
        \node [color=black,inner sep=2,label={[yshift=-0.9cm]$\lone$},circle,fill] at (4,0) () {};
        \node [color=black,inner sep=2,label={[yshift=-0.9cm]$\rone$},circle,fill] at (8,0) () {};
        \draw [color=red!20,line width=2pt] (0,0.25) -- (4,0.25);
        \draw [color=blue!20,line width=2pt] (2,0.5) -- (11,0.5);
\end{tikzpicture}
\caption{Figure depicting the proof that $\out{\ltarget}{\lone}$ is special. The red interval denotes $A$ and the blue interval denotes $B$.}
\label{fig:ltargetlone}
\end{figure}
\end{nestedproof}

Finally, we prove that the pairs of outside atoms such that one is in $\out{\rone}{\rtarget}$ and the other is not in $\out{\lactive1}{\rtarget}$ are covered by at most $5$ cuts in $\cL$.
\begin{claim}\label{lem:right-edges}
There is a collection $\Cover'$ of at most $5$ cuts in $\cL$ that covers every pair of outside atoms such that one is in $\out{\rone}{\rtarget}$ and the other is not in $\out{\lactive1}{\rtarget}$. Furthermore, all the cuts in $\Cover'$ are special except possibly one, which is the set difference of two crossing special cuts.
\end{claim}
\begin{nestedproof}
If $\rtarget=\rone$, no such pairs exist and we are done. Similar to the proof of \Cref{lem:left-edges}, we first show that $\out{\rone}{\rtarget}$ is special. Note that unlike \Cref{lem:left-edges}, this does not immediately finish the proof since we cannot guarantee that $\cut{\rone}{\rtarget}$ is added to $\mathcal{L}$ at depth $\depthact+1$ ($\rtarget$ lies in an active interval starting at $\rone$ at depth $\depthact+1$ and line \ref{line:suffix-chains} of \Cref{alg:construct-laminar} adds cuts from right to left). However, we will see that a similar analysis works to show that we can cover the desired pairs of atoms by adding at most $5$ cuts from depths $\depthact+1$ and $\depthact+2$ in $\cL$ to the collection $\Cover'$.

To show that $\out{\rone}{\rtarget}$ is special, apply \Cref{lem:special-cross-nmc} to the crossing intervals $A = \out{\lactive1}{\rone}$ and $B = \out{\ltarget}{\rtarget}$.
Then either $A \cup B$ or $B \setminus A$ is special; but since $\rone<\rtarget$, $A \cup B = \out{\lactive1}{\rtarget}$ is not special by the definition of $\rone$. 
Hence, $B \setminus A = \out{\rone}{\rtarget}$ is special. \Cref{fig:ronertarget} shows the intervals used in the argument.
\def\scale{0.7}
\begin{figure}[htb!]\centering
\begin{tikzpicture}[scale=\scale]
		\draw [color=black] (-2,0) -- (16,0);
        \node [color=black,inner sep=2,label={[yshift=-0.9cm]$\lactive1$},circle,fill] at (0,0) () {};
        \node [color=black,inner sep=2,label={[yshift=-0.9cm]$\ractive1$},circle,fill] at (14,0) () {};
        \node [color=black,inner sep=2,label={[yshift=-0.9cm]$\ltarget$},circle,fill] at (2,0) () {};
        \node [color=black,inner sep=2,label={[yshift=-0.9cm]$\rtarget$},circle,fill] at (11,0) () {};
        \node [color=black,inner sep=2,label={[yshift=-0.9cm]$\lone$},circle,fill] at (4,0) () {};
        \node [color=black,inner sep=2,label={[yshift=-0.9cm]$\rone$},circle,fill] at (8,0) () {};
        \draw [color=red!20,line width=2pt] (0,0.25) -- (8,0.25);
        \draw [color=blue!20,line width=2pt] (2,0.5) -- (11,0.5);
\end{tikzpicture}
\caption{Figure depicting the proof that $\out{\rone}{\rtarget}$ is special. The red interval denotes $A$ and the blue interval denotes $B$.}
\label{fig:ronertarget}
\end{figure}

Let $\out{\lactive2}{\ractive2}$ be the active interval at depth $\depthact+1$ such that $\lactive2 \le \rtarget \le \ractive2$. Note that $\lactive2 = \rone$ and since $r_S > \rone$, $\lactive2 < r_S \le \ractive2 \le \ractive1$. 
Let 
\begin{align*}\lactive3 &= \max\{ x: \lactive2 \le x \le \rtarget, \out{x}{\ractive2}\text{ is special}\},\\
\ractive3 &= \min\{ x: \rtarget \le x \le \ractive2, \out{x}{\ractive2}\text{ is special}\}\text{, and}\\
\rlast &= \min\{ x: \lactive2 \le x \le \rtarget, \out{x}{\ractive2}\text{ is special}\}.
\end{align*}
If $\lactive3$ does not exist, then there exists an active interval $[\lactive2,r_3']$ at depth $\depthact +2$ such that $\lactive2 \le \rtarget \le r_3'$. Since $\out{\rone}{\rtarget} = \out{\lactive2}{\rtarget}$, the cut $\cut{\rone}{\rtarget}$ will be added to $\mathcal{L}$ by line \ref{line:prefix-chains} of \Cref{alg:construct-laminar} at depth $\depthact+2$ and we are done. Otherwise, $\lactive3, \ractive3, \rlast$ exist. Since $\out{\lactive3}{\ractive2}$, $\out{\ractive3}{\ractive2}$, and $\out{\rlast}{\ractive2}$ are special, line \ref{line:suffix-chains} of \Cref{alg:construct-laminar} adds the corresponding cuts to $\cL$. We add these $3$ cuts to $\Cover'$. The only remaining pairs of outside atoms that are currently uncovered are pairs of outside atoms with:
\begin{enumerate}
\item one in $\out{\rone}{\rlast}$ and the other not in $\out{\lactive1}{\ractive2}$, and
\item one in $\out{\lactive3}{\rtarget}$ and the other in $\out{\rtarget}{\ractive3}$.
\end{enumerate}
\Cref{fig:halfwaypoint} shows the relative positions of the defined points, intervals corresponding to cuts added to $\cL$, and pairs of intervals that are still uncovered.
\def\scale{0.7}
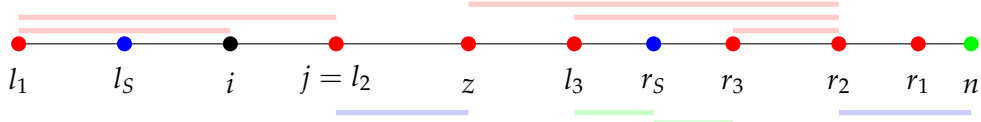
\begin{figure}[htb!]\centering
\begin{tikzpicture}[scale=\scale]
		\draw [color=black] (-2,0) -- (16,0);
        \node [color=red,inner sep=2,label={[yshift=-0.9cm]$\lactive1$},circle,fill] at (-2,0) () {};
        \node [color=red,inner sep=2,label={[yshift=-0.9cm]$\ractive1$},circle,fill] at (15,0) () {};
        \node [color=red,inner sep=2,label={[yshift=-0.9cm]$\ractive2$},circle,fill] at (13.5,0) () {};
        \node [color=red,inner sep=2,label={[yshift=-0.9cm]$\rlast$},circle,fill] at (6.5,0) () {};
        \node [color=red,inner sep=2,label={[yshift=-0.9cm]$\lactive3$},circle,fill] at (8.5,0) () {};
        \node [color=red,inner sep=2,label={[yshift=-0.9cm]$\ractive3$},circle,fill] at (11.5,0) () {};
        \node [color=green,inner sep=2,label={[yshift=-0.9cm]$n$},circle,fill] at (16,0) () {};
        \node [color=blue,inner sep=2,label={[yshift=-0.9cm]$\ltarget$},circle,fill] at (0,0) () {};
        \node [color=blue,inner sep=2,label={[yshift=-0.9cm]$\rtarget$},circle,fill] at (10,0) () {};
        \node [color=black,inner sep=2,label={[yshift=-0.9cm]$\lone$},circle,fill] at (2,0) () {};
        \node [color=red,inner sep=2,label={[yshift=-0.9cm]$\rone=\lactive2$},circle,fill] at (4,0) () {};
        \draw [color=red!20,line width=2pt] (-2,0.25) -- (2,0.25);
        \draw [color=red!20,line width=2pt] (-2,0.5) -- (4,0.5);
        \draw [color=red!20,line width=2pt] (11.5,0.25) -- (13.5,0.25);
        \draw [color=red!20,line width=2pt] (8.5,0.5) -- (13.5,0.5);
        \draw [color=red!20,line width=2pt] (6.5,0.75) -- (13.5,0.75);
        \draw [color=blue!20,line width=2pt] (4,-1.3) -- (6.5,-1.3);
        \draw [color=blue!20,line width=2pt] (13.5,-1.3) -- (16,-1.3);
        \draw [color=green!20,line width=2pt] (8.5,-1.3) -- (10,-1.3);
        \draw [color=green!20,line width=2pt] (10,-1.5) -- (11.5,-1.5);
\end{tikzpicture}
\caption{A figure indicating the relative positions of the defined points and the cuts in $\mathcal{L}$ we considered so far (denoted as red lines). The pairs of atoms between the pair of blue intervals and the pair of green intervals are currently unaccounted for.}
\label{fig:halfwaypoint}
\end{figure}

At depth $\depthact+2$, the intervals $\out{\rone}{\rlast}$ and $\out{\lactive3}{\ractive3}$ are active. Note that $\out{\rone}{\rlast} = \out{\rone}{\rtarget} \setminus \out{\rlast}{\ractive2}$ is the set difference of two crossing special intervals. Thus, line \ref{line:add-active-interval} of \Cref{alg:construct-laminar} adds $\cut{\rone}{\rlast}$ to $\cL$, which covers the pairs of atoms described in (i). We add $\cut{\rone}{\rlast}$ to $\Cover'$.

It remains to show that the pairs of atoms described in (ii), where one is in $\out{\lactive3}{\rtarget}$ and the other is in $\out{\rtarget}{\ractive3}$, are covered by some cut in $\cL$. \nnote{removed a prime.}
If $\lactive3=\rtarget$, we are done, so suppose $\lactive3<\rtarget$. 

By \Cref{lem:special-cross-nmc} with $A = \out{\lactive3}{\ractive2}$ and $B = \out{\ltarget}{\rtarget}$, either $A \cap B$ or $A \setminus B$ is special. Since $\lactive3 < \rtarget$, $A \setminus B = \out{\rtarget}{\ractive2}$ is not special. Hence, $A \cap B = \out{\lactive3}{\rtarget}$ is special.
This means that $\cut{\lactive3}{\rtarget}$ will be added to $\mathcal{L}$ by line \ref{line:prefix-chains} of \Cref{alg:construct-laminar} during the call with active interval $\out{\lactive3}{\rtarget}$, which covers the desired pairs of atoms. We add $\cut{\lactive3}{\rtarget}$ to $\Cover'$. \Cref{fig:lactive3target} shows the intervals used in the argument.
\def\scale{0.7}
\begin{figure}[htb!]\centering
\begin{tikzpicture}[scale=\scale]
		\draw [color=black] (-2,0) -- (16,0);
        \node [color=black,inner sep=2,label={[yshift=-0.9cm]$\lactive2$},circle,fill] at (2,0) () {};
        \node [color=black,inner sep=2,label={[yshift=-0.9cm]$\ractive2$},circle,fill] at (14,0) () {};
        \node [color=black,inner sep=2,label={[yshift=-0.9cm]$\ltarget$},circle,fill] at (0,0) () {};
        \node [color=black,inner sep=2,label={[yshift=-0.9cm]$\rtarget$},circle,fill] at (7,0) () {};
        \node [color=black,inner sep=2,label={[yshift=-0.9cm]$\lactive3$},circle,fill] at (4,0) () {};
        \node [color=black,inner sep=2,label={[yshift=-0.9cm]$\ractive3$},circle,fill] at (10,0) () {};
        \draw [color=blue!20,line width=2pt] (0,0.25) -- (7,0.25);
        \draw [color=red!20,line width=2pt] (4,0.5) -- (14,0.5);
\end{tikzpicture}
\caption{Figure depicting the proof that $\out{\lactive3}{\rtarget}$ is special. The red interval denotes $A$ and the blue interval denotes $B$.}\label{fig:lactive3target}
\end{figure}
\end{nestedproof}
Taking $\Cover$ to be the collection of cuts containing $\cut{\lactive1}{\lone}$, $\cut{\lactive1}{\rone}$, the cut from \Cref{lem:left-edges} and the cuts in $\Cover'$ from \Cref{lem:right-edges} completes the proof. 
\end{proof}

\subsection{Extending the Analysis to Inside Atoms}\label{subsec:general-case}
Next, we show how to extend \Cref{lem:outside-edges-covered} to cover \emph{all} pairs of atoms covered by an $\eta$-near minimum cut, not merely all pairs of outside atoms.

\begin{lemma}\label{lem:inside-edges-covered}
Let $S$ be an $\eta$-near minimum cut with $\eta = \onefortieth$ in a connected component of cuts $\cC$ with $|\cC| > 1$. Let $\cL$ be the laminar family constructed for $\cC$ by \cref{alg:construct-laminar}. Then, there exists a collection of cuts $\Cover \subseteq \cL$ of size at most $8$ such that if $(u,v)$ is a pair of atoms covered by $S$, then $(u,v)$ is covered by $\Cover$.
\end{lemma}
\begin{proof}
Let $\Cover$ be the collection obtained from \Cref{lem:outside-edges-covered}. We show that it covers all pairs of atoms covered by $S$ even if some of them are inside atoms.

Note that cuts corresponding to special intervals can be represented by the intersection of at most $2$ halfplanes bounded by diagonals of the polygon and cuts corresponding to set differences of two crossing special intervals can be represented by the intersection of at most $4$ such halfplanes. Consider the set of diagonals $\cD$ of the polygon corresponding to either $S$ or a halfplane bounding a cut in $\Cover$. Note that $|\cD| \le 1+2 \cdot 7+4= 19$, since by \Cref{lem:outside-edges-covered}, $\Cover$ contains at most $7$ special cuts and one cut corresponding to set differences of two crossing special intervals. This set of diagonals divide the polygon into regions with at most 19 sides (since each diagonal in $\cD$ corresponds to at most one side). By \Cref{thm:halfplanes}, each such region with no outside atoms is empty. Hence, the only inside atoms are in regions containing at least one outside atom. 

Consider a pair of atoms $(u,v)$ covered by $S$. Since each inside atom is in the same region as some outside atom, there exist outside atoms $u',v'$ such that $u,u'$ and $v,v'$ lie in the same region. However, a  diagonal separates $u,v$ if and only if it separates $u',v'$, so a cut in $\Cover$ separates $u,v$ if and only if it separates $u',v'$ (since each cut can be defined by the intersection of some halfplanes defined by $\cD$). In particular, $u',v'$ are outside atoms separated by $S$. By \cref{lem:outside-edges-covered}, $u'$ and $v'$ are separated by some cut in $\Cover$. Hence, $u$ and $v$ are also separated by some cut in $\Cover$. 
\end{proof}
\begin{figure}
    \centering 
\begin{tikzpicture}[scale=0.8]
        \node [fill, circle, color=green!70!black, inner sep=0.5mm] at (0,0) (v) {};
        \node at (0,0.3) () {\footnotesize $v$};
        \node [fill, circle, color=green!70!black, inner sep=0.5mm] at (-1,-1.7) (u) {};
        \node at (-1,-1.4) () {\footnotesize $u$};
        
		\foreach \i/\l in {0/15,18/16,36/17,54/18,72/19,90/0,108/1,126/2,144/3,162/4,180/5,198/6,216/7,234/8,252/9,270/10,288/11,306/12,324/13,342/14,360/15}
			\node [fill, circle, inner sep=0.5mm] at (\i:3) (p_\l) {};
        
		\foreach \i/\j in {0/1, 1/2, 2/3, 3/4, 4/5, 5/6, 6/7, 7/8, 8/9, 9/10, 10/11, 11/12, 12/13, 13/14, 14/15, 15/16, 16/17, 17/18, 18/19, 19/0}
			\path (p_\i) edge (p_\j);

        \path [line width=1.2pt,color=blue] (p_4) edge (p_11);
        \path [line width=1.2pt,color=blue] (p_7) edge (p_11);
        \path [line width=1.2pt,color=blue] (p_11) edge (p_15);
        \path [line width=1.2pt,color=red] (p_6) edge (p_13);
        
        \foreach \i/\l in {207/6,99/0}
			\node [fill,red,circle,inner sep=0.5mm] at (\i:3) (){}; 
        \foreach \i/\l in {207/6,99/0}
			\node [color=red] at (\i:3.5) () {\small $a_{\l}$};

        \foreach \i/\l in {72/-1, 90/0}
			\node  at (\i:3.3) () {\footnotesize $p_{\l}$};
\end{tikzpicture}
\caption{\small The blue lines are the diagonals in $\cD$ and the red line is the representing diagonal of $S$. In the argument of the proof of \Cref{lem:inside-edges-covered}, we have $u'=a_6$ and $v'=a_0$.
}
\label{fig:inside-edges-covered}
\end{figure}
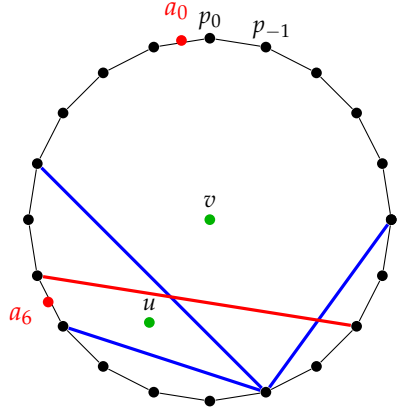
\Cref{lem:inside-edges-covered} finishes the proof of \Cref{thm:L-covers-nmc-general}.

\subsection{Constructing the Full Laminar Family}\label{subsec:full-laminar-family}

Now that we have shown how to construct a laminar family $\cL_P$ on the atoms of a polygon $P$ for a connected component of cuts $\cC$ so that all cuts in $\cC$ are implied by $\cL_P$, we will describe how to construct the overall laminar family. We will begin by constructing a cross-free family $\cR$. 

In particular, initialize a cross-free family $\cR$. Add to $\cR$ one side of each $\eta$-near minimum cut which is not crossed (i.e. every connected component $\cC$ with $|\cC|=1$). Then, construct a polygon $P$ for each connected component $\cC$ of near minimum cuts with $|\cC| > 1$. Finally, run our construction procedure on each polygon $P$ (setting one of the outside atoms as root), and add the resulting laminar family of cuts $\cL_P$ to $\cR$. 

Here we make use of the following from Bencz{\'u}r:
\begin{lemma}[{\cite[Lemma 4.1.7]{Ben97}}]\label{lem:root_atom} Let $\cC,\cC'$ be two distinct connected components of crossing cuts for a family of cuts of $G=(V,E)$. Then, there exists an atom $a \in \cA(\cC)$ and $a' \in \cA(\cC')$ so that $a \cup a' = V$. 
\end{lemma}

\begin{lemma}\label{lem:cross-free}
    $\cR$ is a cross-free family.
\end{lemma}
\begin{proof}
    Every cut added to $\cR$ is the union of some set of atoms in $\cA(\cC)$ for some connected component $\cC$ (note if $|\cC|=1$ then $\cA(\cC) = C$ for the unique $C \in \cC$). 

    By construction, every $S,T$ over the same set of atoms $\cA(\cC)$ do not cross. So, suppose $S$ is the union of a set of atoms in $\cA(\cC)$ and $T$ is the union of a set of atoms in $\cA(\cC')$ for $\cC \not= \cC'$. By \cref{lem:root_atom}, there is an atom $a \in \cA(\cC)$ and $a' \in \cA(\cC')$ so that $a \cup a' = V$. Recall that $\cA(\cC)$ and $\cA(\cC')$ both form partitions over the vertex set.
    \begin{itemize}
        \item If $a \subseteq S$ and $a' \subseteq T$, then $S \cup T = V$.
        \item If $a \subseteq S$ and $a' \cap T = \emptyset$, then $T \subseteq S$. 
        \item If $a \cap S = \emptyset$ and $a' \subseteq T$, then $S \subseteq T$. 
        \item If $a \cap S = \emptyset$ and $a' \cap T = \emptyset$, then $S \cap T = \emptyset$.
    \end{itemize}
    In each possible case they do not cross, so $\cR$ is cross-free. 
\end{proof}

Finally we will let $\cL$, our final laminar family, be the result of fixing an arbitrary node $r \in V$ and for each $S \in \cR$, adding whichever of $S$ and $V \setminus S$ does not contain $r$. 
Given this, we can prove our main theorem.
\mainthm*
\begin{proof}
    We first compute the set of $\eta$-near minimum cuts of $G$, and for each connected component $\cC$, we compute the polygon representation. This can be done in polynomial time \cite{Ben95}. For each connected component (with given polygon representation $P$) we apply \cref{thm:L-covers-nmc-general} and produce a laminar family $\cL$ as described above using \cref{lem:cross-free} so that every cut $S \in \cC$ is covered by at most 8 cuts in $\cL$. 

    Then, apply \cref{thm:strong-tree-laminar} to $\cL$ to find a spanning tree $T$ which has at most $11$ edges in every cut in $\cL$ and return $T$.

    Fix an $\eta$-near minimum cut $S$. If it does not lie in a connected component of size at least 2, then $|T \cap \delta(S)| \le 11$ immediately since $S \in \cL$. Otherwise, $S$ lies in some polygon $P$. By \cref{thm:L-covers-nmc-general}, the edges of $S$ are contained in the edges of at most 8 cuts from $\cL$. So, $|T \cap \delta(S)| \le 88$. 
\end{proof}

\paragraph{Acknowledgments.} Initial discussions on this topic took place during the Trimester Program on Combinatorial Optimization at the Hausdorff Research Institute for Mathematics; we are grateful to the institute for its support.

\printbibliography

\appendix
\section{Corollaries from Prior Work}\label{appendix}

Recall the theorem mentioned in the introduction:
\thmthin*
As mentioned, this is not exactly the statement from \cite{KO23}, and further the second statement for laminar families requires some small additional argument. 
We provide these now.

Klein and Olver~\cite{KO23} show that given a feasible solution $x \in \R_{\ge 0}^E$ to the LP relaxation which requires that $x$ is in the spanning tree polytope and $x(\delta(S)) \le b_S$ for integral bounds $b_S$, one can round $x$ to a spanning tree that violates the bounds by at most a factor of 22. One can exhibit a point in the dominant of this polytope in a general $k$-edge-connected graph by setting $x_e = \frac{2}{k}$ for all edges, thus giving $b_S = \lceil \frac{2}{k}|\delta(S)| \rceil$ and an overall violation of $22 \cdot \lceil \frac{2}{k}|\delta(S)| \rceil \le 22 \cdot \frac{2}{k}|\delta(S)| + 22 \le \frac{66}{k}|\delta(S)|$, where we used that $|\delta(S)| \ge k$. 

Next we consider the case where  $|\delta(S)| \leq \tfrac43k$ for all $S \in \cL$.
\cite{KO23} call a point $x$ \emph{$\cL$-aligned} if $x(E(S)) = |S|-1$ for all $S \in \cL$, where $E(S)$ is the set of edges with both endpoints in $S$. 
For $\cL$-aligned points, they give a guarantee of $2 x(\delta(S)) + 3$. Shortly, we will show that the vector $\bar{x} \in \R_{\ge 0}^E$ defined by $\bar{x}_e = \frac{3}{k}$ for all edges dominates a point $x'$ which is $\cL$-aligned. Using that $|\delta(S)| \leq \frac43k$, this allows us to  obtain a guarantee of
    $$2 \cdot \frac{3}{k}|\delta(S)| + 3 \le 11.$$
    
    Now we show that $\bar{x}$ dominates an $\cL$-aligned point $x'$. In particular, it suffices to show that $\bar{x}$ is in the dominant of the spanning tree polytope of $G[S]$ for all $S \in \cL$ (see \cite{KO23} for further details). The standard formulation of the dominant of the spanning tree polytope requires that for any partition $\Pi$ of the vertex set, we have $\sum_{P \in \Pi} \bar{x}(\delta(P)) \ge 2(|\Pi|-1)$. Now, let $P_1,\dots,P_\ell$ be any partition of a set $S \in \cL$. Then
    $$\sum_{P \in \Pi} \bar{x}(\delta(P)) \ge 3|\Pi| - \bar{x}(\delta(S)),$$
    where we use that $\bar{x}(\delta(P)) \ge 3$ for all partitions, as $G$ is $k$-edge-connected. Since $\bar{x}(\delta(S)) = \tfrac3k |\delta(S)| \le \frac{4}{3}k \cdot \frac{3}{k}$, and $|\Pi| \geq 2$, we finish the proof by noticing:
    $$\sum_{P \in \Pi} \bar{x}(\delta(P)) \geq 3|\Pi| - 4 \ge 2|\Pi| - 2 = 2(|\Pi|-1).$$

\nnote{I removed the remark; if we put it back, maybe say the strengthening more clearly?}
\end{document}